\documentclass[preprint,12pt]{elsarticle}
\usepackage{amssymb,amsfonts,amsmath,amsthm}
\usepackage[boxed,linesnumbered]{algorithm2e}
\usepackage{citet,booktabs}
\usepackage{tikz}
\usepackage{comment}

\usetikzlibrary{shapes}
\usetikzlibrary{calc}

\DontPrintSemicolon

\newtheorem{Theorem}{Theorem}
\newtheorem{Lemma}{Lemma}
\newtheorem{Claim}{Claim}
\newtheorem{Fact}{Fact}

\newcommand{\ceil}[1]{\lceil #1 \rceil}
\newcommand{\FACTOR}{\textsf{Con-$f$-F}}

\title{A Characterization for the Existence of Connected $f$-Factors of {\em Large} Minimum Degree\tnoteref{fn1}}
\begin{document}
\author[nsn]{N.~S.~Narayanaswamy\corref{cor1}}
\ead{swamy@cse.iitm.ac.in}
\author[nsn]{C.~S.~Rahul\corref{cor1}}
\ead{rahulcs@cse.iitm.ac.in}
\address[nsn]{Indian Institute of Technology Madras, Chennai, India 600036}
\tnotetext[fn1]{Supported by the Indo-German Max Planck Center for Computer Science grant for the year 2013-2014 in the area of Algorithms and Complexity}
\begin{abstract}
It is well known that when $f(v)$ is a constant for each vertex $v$, the connected $f$-factor problem is NP-Complete.   
In this note we consider the case when  $f(v) \geq \lceil \frac{n}{2.5}\rceil$ for each vertex $v$, where $n$ is the number of vertices. 
We present a diameter based characterization of graphs having a connected $f$-factor (for such $f$).  We show that if a graph $G$  has a connected $f$-factor and  an $f$-factor with 2 connected components, then it has a connected $f$-factor of diameter at least 3.  This result yields a polynomial time algorithm which first executes the Tutte's $f$-factor algorithm, and if the output has 2 connected components, our algorithm searches for a connected $f$-factor of diameter at least 3.
\end{abstract}
\maketitle
\section{Introduction}
\label{sec:intro}
Given a simple undirected graph $G=(V,E)$, (see \citet{DW00}), on $n$ vertices, an $f$-factor  of $G$ is a spanning subgraph $H$ such that $d_H(v)=f(v)$ for every $v$ in $V$. The problem of deciding whether a given graph $G$ has an $f$-factor is a well studied problem over many years by \citet{An85,tutte1952factors,cornuejols1988general, liu2008toughness} and \citet{PM07}. \citet{WT54} has shown the problem to be polynomial time reducible to perfect matching.  Since perfect matching is polynomial time solvable it follows that testing for an $f$-factor is polynomial time solvable, and we refer to this algorithm as the {\em Tutte's $f$-factor algorithm}.
We consider the problem of finding a connected subgraph which is an  $f$-factor.
   For the case where $f(v) = 2$ for  every $v$ in $V$, a connected $f$-factor is a Hamiltonian cycle. It is well known that the Hamiltonian cycle problem is NP-Complete. In fact, \citet{CC90} have shown that for each constant $d$, when $f(v)=d$ for every $v$, the connected $f$-factor problem is NP-Complete. We refer this connected regular spanning subgraph as a connected $d$-factor. Further, recently \citet{BN13}  showed that for  every $0<\epsilon<1$, for $f(v) \geq n^{1-\epsilon}$ the connected $f$-factor problem is NP-Complete. 

 For the case when $f(v) \geq \ceil {\frac{n}{2}} -1$,  \citet{BN13}  have also shown a simple test for the existence of a connected $f$-factor. 
We observe that when $f(v)\geq \ceil{n/2}$ for every $v$ in $V$, any $f$-factor of a given graph $G$ is connected and has diameter at most 2. From the work by \citet{BN13} when $f(v)=\ceil{n/2}-1$, if the input graph has a $f$-factor with two components and {\em also} a connected $f$-factor, then there exists a connected $f$-factor of diameter 3. Similarly, for the case where $f(v) \geq \frac{n}{c}$ for each $v$ in $V$ and a constant $c$, the diameter of a connected $f$-factor is upper bounded by $3c-1$. This is because if we consider a diametric path $u_1,u_2,\cdots,u_{3c}$, the sets $N(u_i)$ and $N(u_{i+j})$ have to be disjoint for each $j \geq 3$. Further, each such set $|N(u_i)| \geq \frac{n}{c}$ and it is impossible to have such $c$ pairwise disjoint vertex sets. Thus for $f(v) \geq \frac{n}{c}$ for each $v$ in $V$, there seems to be a concrete  correspondence between the diameter of connected $f$-factors and the value $c$. We use this correspondence to come up with a polynomial time algorithm for $c=2.5$. It may be possible to extend the idea to solve the problem for any constant $c$. 


\noindent
{\bf Our Work.}
We present a diameter based characterization of graphs having connected $f$-factors where  $f(v) \geq\ceil{\frac{n}{2.5}}$ for every $v$ in $V$. For each such $f$, we refer to the connected $f$-factor problem as $\FACTOR$.   Our main result which we prove in Section \ref{sec:method} is as follows:
\begin{Theorem}
 Let $G$ be an instance of \FACTOR. If each connected $f$-factor of $G$ has  diameter at most 2, then every $f$-factor of $G$ is connected.
 \label{mainthm}
\end{Theorem}
\noindent
We use the above characterization to show that $\FACTOR$ can be solved in polynomial time using  the following two main graph theoretic results.
\begin{enumerate}
\item If an instance $G$ of $\FACTOR$ has a connected $f$-factor and also an $f$-factor $H$ with two components , then $G$ has a connected $f$-factor $G'$ of diameter at least 3. 
\item Given two vertices $u,v$  in $V(G)$, the problem of computing an $f$-factor in which distance between $u$ and $v$ is 3, reduces to perfect matching. This result can be seen as a generalization of the Tutte's reduction of the $f$-factor problem to the  graph perfect matching problem.  
\end{enumerate}

\noindent
{\bf Preliminaries and Notations}\\
Throughout this paper, $f$ is a function $f:V \rightarrow \mathbb{N}$ such that for each $v$ in $V$, $f(v) \geq \ceil{\frac{n}{2.5}}$ where $n=|V|$. Consequently, we have the following fact.
\begin{Fact}
Let $G_f$ be an $f$-factor of $G$. Then the number of components in $G_f$ is at most 2.
\end{Fact}
\noindent
Unless otherwise mentioned, $G$ always represents the input graph. 
We also assume that the number of vertices in $G$, denoted by $n$, is at least 12 (the cause for this will be clear from the analysis).
We use the following standard definitions and notations from \citet{DW00}: degree $d_G(v)$ of a vertex $v$ in a graph $G$, minimum degree $\delta(G)$,  and the open neighborhood, $N(v)$ of a vertex $v$. We leave out the subscript $G$ when the graph in question is clear from the context.  Further, the concepts of bridge or cut-edge, the edge-cut $[X,V \setminus X]$ created by a vertex partition $\{X, V \setminus X\}, X \subseteq V$, a tour, the diameter $diam G$, a circuit,  and the subgraph of G induced by $S \subseteq V$, denoted by $G[S]$, are all standard.   For a set $S \subseteq V$, $\displaystyle N(S) = \bigcup_{v \in S} N(v) \setminus S$ is the open neighbourhood of  $S$.
Let $\{X,V \setminus X\}$ be a partition of $V$. We say $E' \subseteq E(G)$ is \textit{incident on} $\{X,V \setminus X\}$ if $E' \cap [X, V \setminus X] \neq \emptyset$.  Further, we say $E' \subseteq E(G)$ \textit{covers} $\{X,V \setminus X\}$ if,  $E'$ is incident on $\{X, V \setminus X\}$, and 
there exists $S \in \{X,V \setminus X\}$ such that $\forall v \in S$, $v$ is incident on some edge $e \in E'$.\\

\noindent
{\bf Equitably Colored Graphs and Alternating Circuits}  were introduced by \citet{DP83}.  Let $G$ be a edge colored graph in which each edge is colored either red or blue.  $G$ is said to be equitably colored if for every vertex $v \in V(G)$, the number of red edges incident on $v$ is equal to the number of blue edges incident on $v$.    An eulerian closed trail $v_1,v_2, \ldots, v_{2t},v_1$ (a sequence of vertices) is said to be an alternating circuit if the edges $\{v_1,v_2\}, \{v_2,v_3\}, \ldots, \{v_{2t},v_1\}$ are coloured alternatingly red and blue.   In other words, an alternating-circuit  is a circuit that is equitably colored. A  minimal alternating circuit $T$ is an alternating circuit such that no proper subset of edges of $T$ forms an alternating circuit.  

\noindent
Let $G'$  be a spanning subgraph of an edge colored graph $G$. An alternating circuit $T$ in $G$ is defined to be a \textit{switch on} $G'$ if the set of red edges in $T$  is a subset of $E(G')$ and the set of the blue edges of $T$ is disjoint from $E(G')$.  If $T$  is a switch on $G'$, the  operation Switching($G'$,$T$)  is a graph obtained by removing from $G'$ all the red edges in $T$  and adding all the blue edges in $T$ to $G'$.

\section{On the Diameter of Connected $f$-factors}
\label{sec:method}
\noindent
In this section we present our  diameter based characterization for graphs which have a connected $f$-factor.
  The following two facts are simple graph theoretic facts and the proofs are straightforward. 
\begin{Fact}
Let $G$ be an undirected graph such that the minimum degree $\delta(G) \geq \frac{n}{2}$, then the diameter of $G$ is at most 2.
 \label{diameter2}
\end{Fact}
\begin{Fact}
Let $G$ be a graph in which each edge is assigned a color from the set $\{red,blue\}$.  $G$ is a set of vertex  disjoint  alternating-circuits if and only if  $d_R(v)=d_{B}(v)$ for all  $v$ in $G'$. \label{rblemma}
\end{Fact}
\begin{Lemma} 
Let $T$ be a minimal alternating-circuit.  For each vertex $v$, both  $d_{R}(v)= d_{B}(v) \leq 2$.
\label{rbminimalitythm}
\end{Lemma}
\begin{proof}
The proof is by contradiction. Let $v$ be a vertex in $T$ which has at least 3 incident red edges.  We show the existence of  an alternating circuit $T'$ whose edges are a proper subset of the edges of $T$ as follows: Consider the sequence of edges  obtained from $T$ by starting at the first occurence of $v$ in $T$, and returning to $v$ for the {\em the second time}.   Let $e_1, e_2, e_3, e_4$ be the edges incident on $v$ in the sequence, in the same order in which they were present in the sequence.   $e_2$ and $e_3$ are consecutive in the sequence and are of different colors.  If $e_1$ and $e_2$ are of different colors, then the sequence of edges from $e_1$ to $e_2$ is an alternating circuit, thus contradicting the minimality of $T$.  If $e_1$ and $e_2$ are of the same color, and if $e_3$ and $e_4$ are of the same color,  all the edges in the sequence form an alternating circuit, and is proper subset of the edges in $T$.  The proper containment follows because at least one red edge incident of $v$ is not present in this sequence.  This contradicts the minimality of $T$.  If $e_3$ and $e_4$  are of different colors, then the sequence of edges from $e_3$ to $e_4$ is an alternating circuit, again contradicting the minimality of $T$.  Hence the lemma.

\end{proof}
\begin{Lemma}
\label{cutincidencelemma}
Let $G$ be a graph of diameter 2 and let $\{X,V \setminus X\}$ be a partition of $V(G)$. There exists a set of edges $E' \subseteq [X,V\setminus X]$ that covers $\{X, V \setminus X\}$.
\end{Lemma}
\begin{proof}
We obtain a contradiction to the fact that the graph has diameter 3 by assuming that the lemma is false. 
Let  $v_1 \in X, v_2 \in V\setminus X$ such that $N(v_1) \subseteq X$ and $N(v_2) \subseteq V \setminus X$.  
Then the shortest path between $v_1$ and $v_2$ is of length at least 3, and this contradicts the premise $diam ~G = 2$.  Hence the claim in the lemma is true.
\end{proof}
\noindent
In the following,  given an $f$-factor of $G'$ of $G$ we consider the coloring of $G$ in which $E(G')$ is colored red, and $E(G) \setminus E(G')$ is colored blue.  The alternating circuits we consider are with respect to this coloring.
\begin{Lemma}
Let $G$ be a graph with at least 12 vertices.  
 Let $X$ and $V \setminus X$ be two components of an $f$-factor $G'$ of $G$.   Let $T$ be a minimal alternating-circuit which is a switch on $G'$ and is incident on $\{X, V \setminus X\}$.  
Then Switching($G'$,$T$) is a connected $f$-factor of $G$.
\label{minauglemma}
\end{Lemma}
\begin{proof}
Since $X$ and $V \setminus X$ have at least $\ceil{\frac{n}{2.5}}+1$ vertices, it follows that both $X$ and $V \setminus X$ are of size at most $\lfloor{\frac{1.5n}{2.5}}\rfloor-1$. Let $G''$ denote the graph Switching($G'$,$T$). Without loss of generality, let us consider $X$. 
We know that the minimum degree in $G'[X]$ is at least $\ceil{\frac{n}{2.5}}$. Since  $T$ is a minimal alternating circuit,  from Lemma \ref{rbminimalitythm}, it follows that the degree of each vertex in $G''[X]$ is at least $\ceil{\frac{n}{2.5}}-2$.  Further, the number of vertices in $X$ is at most $\lfloor{\frac{1.5n}{2.5}}\rfloor-1$.  For $n\geq 12$, $\lfloor{\frac{1.5n}{2.5}}\rfloor-1< 2 \times (\ceil{\frac{n}{2.5}}-2)$. In other words the minimum degree in $G''[X]$ is more than half the number of vertices in $X$.  Consequently, by Fact  \ref{diameter2}, $G''[X]$ is of diameter at most 2.  Similarly, $G''[V \setminus X]$ is of diameter at most 2. Further, since $T$ is incident on $\{X,V \setminus X\}$, there is an edge with one end point in $X$ and the other in $V \setminus X$.   Therefore $G''$  is connected.   
\end{proof}
%
\noindent
The next lemma plays  a critical role in the polynomial time computability of \FACTOR.  It implies that if Tutte's $f$-factor algorithm returns a solution with 2 components and if the graph has a connected $f$-factor, then there is one of diameter at least 3.
\begin{Lemma}
Let $G'$ be an $f$-factor of $G$ with two components $\{X, V\setminus X\}$. 
If $G$ has a connected $f$-factor, then there exists a minimal alternating-circuit $T$ incident on $\{X, V \setminus X\}$ and is a switch on $G'$. Further, Switching($G'$,$T$) is an $f$-factor of diameter at least 3.
\label{switchinglemma} 
\end{Lemma}
\begin{proof}
We prove by contradiction on the diameter of the $f$-factor computed by the switching operation.  Let $H$ be a connected $f$-factor of $G$. Color the edges in $G'$ with color red and those in $H$ with color blue. By Fact \ref{rblemma}, each component in $E(G')\bigtriangleup E(H)$
is an alternating-circuit. Since $H$ is connected, there exists a minimal alternating-circuit $T$ in $E(G')\bigtriangleup E(H)$, which is a switch on $G'$ and is incident on $\{X,V\setminus X\}$. From Lemma \ref{minauglemma} it follows that Switching($G'$,$T$), denoted by $G''$, is a connected $f$-factor.   Our proof analyzes two cases on the structure of such  a $T$:\\ 
{\bf Case 1:  There exists a $T$ with exactly four edges:}In this case we prove that   $G''$ is of diameter at least 3.  Let $u \in X$ and $v \in V \setminus X$ be vertices not in $V(T)$.  Clearly, these are at a distance of atleast 3 in $G''$.  \\
{\bf Case 2:  Every $T$ is of length more than 4:}  Let us assume that $G''$ is of diameter at most 2.  Consider an edge $(u_1,u_2)$ in $E(G'')$ such that $u_1 \in X$ and $u_2 \in V \setminus X$. Now, the number of vertices 
in $N_{G'}(\{u_1,u_2\})$ is at least $\ceil{\frac{2n}{2.5}}$, since $G'$ is an $f$-factor.  Therefore the number of vertices in $V \setminus N_{G'}[\{u_1,u_2\}]$ is at most $\lfloor \frac{n}{5} \rfloor-2$. Since we have assumed that $G''$ is of diameter at most 2,  from Lemma \ref{cutincidencelemma} we know that one of the two sets from $\{X, V \setminus X\}$ has the property that each vertex in the set is incident on an edge in $[X, V \setminus X] \cap E(T)$.  Without loss of generality, let $X$ be this set and let $E' = [X, V \setminus X] \cap E(T)$.  Therefore, each vertex in $N_{G'}(u_1)$ is incident on an edge in $[X,V \setminus X] \cap E(T)$.  Further, each  edge incident on a vertex in $N_{G'}(u_1)$ is not incident on a vertex in $N_{G'}(u_2)$- the reason for this is that the existence of such an edge will result in an alternating circuit consisting of 4 edges, and we are in the case where such an alternating circuit does not exist.    Therefore, the number of blue edges in $E'$ incident on vertices in $(V \setminus X) \setminus N_{G'}(u_2)$ is at least $\ceil{{\frac{n}{2.5}}}$.  Since $T$ is a minimal alternating circuit, $u_2$ has at most one more blue edge incident on it other than $\{u_1,u_2\}$.  Therefore, it follows that the number of  edges of $E'$ incident on vertices in $(V \setminus X) \setminus N_{G'}[u_2]$ is at least $\ceil{{\frac{n}{2.5}}}-1$.  On the other hand, since $u_1$ and $u_2$ are in different components in $G'$, $|(V \setminus X) \setminus N_{G'}[u_2]| \leq |V \setminus N_{G'}[\{u_1,u_2\}]|$. Since $|V \setminus N_{G'}[\{u_1,u_2\}]| \leq \lfloor \frac{n}{5} \rfloor-2$, it follows that $|(V \setminus X) \setminus N_{G'}[u_2]| \leq\lfloor \frac{n}{5} \rfloor-2$.
Since $T$ is a minimal alternating-circuit, by Lemma \ref{rbminimalitythm} the blue degree at each vertex  is at most 2.
Therefore, the number of edges in $E'$ incident on the vertex set $(V \setminus X) \setminus N_{G'}[u_2]$ is at most twice the size of the set.   Thus the number of edges in $E'$ incident on  $(V \setminus X) \setminus N_{G'}[u_2]$ is at most $\lfloor \frac{n}{2.5} \rfloor -4$. 
We have derived a   contradictory set of inequalities involving the number of edges in $E'$ incident on $(V \setminus X) \setminus N_{G'}[u_2]$.
Therefore our assumption on the diameter of $G''$ being at most 2 is wrong.   Hence the lemma. 
\end{proof}
\noindent
We now complete the proof of Theorem \ref{mainthm}.\\
\textit{\textbf{Proof of Theorem \ref{mainthm}:}}
\begin{proof}
 If $G$ has  an  $f$-factor $G'$ with two components and  also has a connected $f$-factor,  then, from Lemma \ref{switchinglemma}, there exists a minimal alternating-circuit $T$ such that $G''$=Switching($G'$,$T$) is a connected $f$-factor of diameter at least 3.  Consequently, it follows that if all
connected $f$-factors of $G$ are of diameter at most 2, then all $f$-factors of $G$ are connected.  Hence the theorem.
\end{proof}
\noindent
In the next section, we use this characterization to design a polynomial time algorithm for $\FACTOR$. 
\section{Polynomial time Algorithm for the Connected $f$-factor Problem}
We start by presenting the Algorithm $A_f$ for the connected $f$-factor problem as follows:
\begin{enumerate}
\item Run Tutte's $f$-factor algorithm with $G$ and $f$ as input.
\item {\bf If} the algorithm fails to return an $f$-factor, then exit after reporting Failure.
\item {\bf Else} If it returns a connected $f$-factor, then  output it and exit.
\item {\bf Else} Let $G'$ be the output $f$-factor, with two components $\{X, V \setminus X\}$.
\begin{enumerate}
\item For each $u \in X$ and $v \in V \setminus X$, and each induced path $P_{uv}$ of 3 edges \\
{\tt/*  Lemma \ref{diameterlemma}  shows that one of these choices is correct. */}
\begin{enumerate}
\item Test for an $f$-factor of $G$ containing $P_{uv}$ as a shortest path between $u$ and $v$. \\
{\tt /*  Using  subroutine Distance-Constrained-Factor */}
\item If an $f$-factor is found, output it and exit. \\
{\tt /* Theorem \ref{3diafact} shows that it is a connected $f$-factor */}
\end{enumerate}
\item Exit reporting failure.
\end{enumerate}
\end{enumerate}
We start with the following lemma that is crucial to prove the correctness of the iteration in Step 4.a and then present a variant of Tutte's reduction to deal with Distance Constrainted $f$-factors.
\begin{Lemma}
Let $G$ be a graph with at least 12 vertices, and let $G'$ be an $f$-factor of $G$ with two components $\{X, V\setminus X\}$.  If $G$ has a connected $f$-factor $G''$=Switching($G'$,$T$) for a minimal alternating-circuit $T$, then 
for each pair of vertices $u,v$  at a distance at least 3 in $G''$,  exactly one of $u$ and $v$ is in $X$.
\label{diameterlemma} 
\end{Lemma}
\begin{proof}
Since $G'$ is an $f$-factor, it follows that $|X| < \frac{3n}{5}$ and $ |V \setminus X| < \frac{3n}{5}$.  Since $f(v) \geq \frac{n}{2.5}, v \in V$, it follows from Fact \ref{diameter2} that the diameter of $G'[X]$ and $G'[V \setminus X]$ is at most 2.
Further  by Lemma \ref{rbminimalitythm}, we know that the minimum vertex degree in  $G''[X]$ and $G''[V \setminus X]$ is at least $\frac{n}{2.5}-2$- since $G''=$Switching($G'$,$T$) where $T$ is a minimal alternating-circuit.  Since $n \geq 12$, $\frac{n}{2.5}-2 > \frac{1}{2} \times \frac{2n}{5}$.  
For $n \geq 12$, in $G''$, the minimum degree in  $G''[X]$ is at least $\frac{|X|}{2}$ and the minimum degree  in $G''[V \setminus X]$ is at least $\frac{|V \setminus X|}{2}$.  Therefore, by Fact \ref{diameter2}, $G''[X]$ and $G''[V \setminus X]$ are subgraphs of diameter at most 2.    Therefore, any two vertices $u$ and $v$ which are at distance at least 3 cannot  both be in $X$ or $V \setminus X$.  Hence the lemma. 
\end{proof}
\begin{algorithm}[H]
\label{Algo2}
\SetKwInOut{Input}{Input}\SetKwInOut{Output}{Output}
\Input{$G(V,E)$,$f$,$u$,$v$}\Output{$H$ an undirected graph}
Remove edge $\{u,v\}$ from $G$ if exists.\\
{\tt /*  Tutte's Reduction from $f$-factor to Perfect Matching */}\\
For each vertex $x$ in $V(G)$, Let $e(x) = d_G(x) - f(x)$; \\
Construct  graph $H$ as follows:\\
For each $x$ in $V(G)$, add a biclique $K_{\{e(x),d(x)\}}$, having  sets $A(x)$ of size $d(x)$ and $B(x)$ of size $e(x)$.\\
For each edge $(q,w)$ in $E(G)$, add an edge involving one vertex of $A(q)$ and one vertex of $A(w)$, and for each $\upsilon \in V$, each vertex of $A(\upsilon)$ participates in one such edge. \\
{\tt /* Enforcing the distance of more than 2 between $u$ and $v$*/ }\\
Let $S=\{l| (u,l) \text{ and } (l,v) \in E(G)\}$. //{\tt the $f$-factor should avoid each path  $P_{uv}$ of length two}\\
For each $l \in S$ steps 10 and 11:\\
\hspace*{0.5cm}Let $\{x,y\} \subseteq A(l)$ such that  $x$ is adjacent to a vertex in $A(u)$ 
      and $y$ is adjacent to a vertex in $A(v)$ in $H$.\\
\hspace*{0.5cm}Let $z \in B(l)$. Remove edge set  $\{\{z,w\}|w \in A(l)\setminus \{x,y\}\}$ from $H$.\\ 
\caption{Redn-PM($G$,$f$,$u$,$v$)-
Reduction from $f$-factor with Distance Constraints to Perfect Matching}
\end{algorithm}
\noindent
The routine {\bf Distance-Constrained-Factor($G$,$f$,$u$,$v$)} performs the following steps to compute an $f$-factor $G'$ in which the distance between $u$ and $v$ is at least 3.
\begin{enumerate}
\item Generate an instance of perfect matching $H$ using the reduction Redn-PM($G$,$f$,$u$,$v$) as shown in Algorithm 1.
\item If $H$ does not have a perfect matching report failure to find an $f$-factor in which distance between $u$ and $v$ is at least 3 and exit.
\item Let $M$ be a perfect matching in $H$. Compute the $f$-factor $G'$ from $M$ and $H$ using the $f$-factor computation in Tutte's reduction.
\end{enumerate}
We state the following theorem without proof about the reduction described in Algorithm 1.   
The proof  has been left out as it is exactly on the lines of Tutte's proof as presented in the book by \citet{DW00} (page 141).
\begin{Theorem}
\label{tutteredn}
Let $G(V,E)$ be an undirected graph, function $f:V \rightarrow \mathbb{N} \cup \{0\}$ and vertices $u,v \in V$.  $G$ has an $f$-factor in which the distance between $u$ and $v$ is at least 3 if and only if $H$ output by Redn-PM($G$,$f$,$u$,$v$) has a perfect matching.
\end{Theorem}
\noindent
\begin{Theorem}
Let $G\in$ \FACTOR. Then algorithm $A_f$ outputs a connected $f$-factor of $G$ in polynomial time.
\label{3diafact}
 \end{Theorem}
\begin{proof}
 If $G$ does not have a connected $f$-factor, then the algorithm reports this correctly.  
 If all the $f$-factors are connected $f$-factors, then the algorithm is correct, as the first step is the Tutte's $f$-factor algorithm, and it will find a connected $f$-factor. On the other hand if there is an $f$-factor with two components and $G$ has a connected $f$-factor, then we know from Lemma \ref{switchinglemma}  that there exists a connected $f$-factor $G''$ = Switching($G'$, $T$) of diameter at least 3, where $T$ is a minimal alternating-circuit. Further, from Lemma \ref{diameterlemma} we know that any pair of vertices $u,v$ such that distance in $G''$ is 3 are not both in the same component.
 The algorithm $A_f$ enumerates each such candidate path $P$ of length 3 for each pair of vertices, one from $X$ and the other from $V \setminus X$, and, using Distance-Constrained-Factor, checks for an $f$-factor of $G$ containing $P$ as a  shortest path between $u$ and $v$.   To check for an $f$-factor of $G$ containing $P$ as a shortest path, we reduce the $f$ value of $u$ and $v$ by 1,  the values of $a$ and $b$ by 2, remove the edges which are among the vertices of $P$ in $G$, and invoke Distance-Constrained-Factor. If for a path $P=\{u,a,b,v\}$ of length 3 between two vertices $u$ and $v$, an $f$-factor is found by Distance-Constrained-Factor, then we claim that this $f$-factor is indeed a connected $f$-factor in $G$. The proof is by contradiction - Let us assume that the graph is not connected.  Let $Y$ be the component that contains $u$ and $v$ and the path $P$.  Since $P$ is a shortest $u$-$v$ path in the output factor, it follows that $N(v) \cap N(u) = \emptyset$.  Therefore, the number of vertices in the component $Y$ is at least $\frac{2n}{2.5}+2$.  
The number of vertices $V \setminus Y$ is at most $n - \frac{2n}{2.5}-2= \frac{n}{5}-2$. However, this is a contradiction to the fact that in an $f$-factor, each component  has at least $\frac{n}{2.5}+1$ vertices.  Therefore, the $f$-factor found by the algorithm is indeed a connected $f$-factor and consequently, the algorithm is correct.  
Given $G'$, there are at most $n^4$ such paths to enumerate, and Distance-Constrained-Factor is basically a polynomial time perfect matching computation due to Theorem \ref{tutteredn}.  Therefore, $A_f$ runs in polynomial time and decides whether $G$ has a connected $f$-factor. 
\end{proof}

\bibliographystyle{plainnat}
\bibliography{dregular}
\end{document}